\documentclass[12pt]{iopart}

\usepackage{inputenc}
\usepackage{hyperref}
\usepackage[T1]{fontenc}
\expandafter\let\csname equation*\endcsname\relax
\expandafter\let\csname endequation*\endcsname\relax
\usepackage{amsmath}
\usepackage{amsthm}
\usepackage{amssymb}
\usepackage{amsfonts}
\usepackage[english]{babel}
\usepackage{latexsym}
\usepackage{ulem}
\usepackage{xcolor}
\usepackage{txfonts}
\usepackage{pxfonts}
\usepackage{dsfont}
\usepackage{mathrsfs}
\usepackage{enumitem}
\usepackage{indentfirst}

\newcommand{\secref}[1]{(\ref{#1})}

\newcommand{\F}{\mathscr{F}}
\newcommand{\W}{\mathscr{W}}
\newcommand{\M}{\mathcal{M}}
\newcommand{\N}{\mathcal{N}}

\newcommand{\comm}[2]{\left[#1,#2\right]}

\newtheorem{mydef}{Definition}[subsection]

\newtheorem{Theorem}{Theorem}[subsection]
\newtheorem{ATheorem}{Theorem}[section]

\begin{document}

\title{``Vacuum-like'' Hadamard states for quantum fields on curved spacetimes}
\author{Marcos Brum$^{1,2}$, Klaus Fredenhagen$^{2}$}
\address{$^{1}$ Instituto de F\'{\i}sica, Universidade Federal do Rio de Janeiro, Caixa Postal 68528, Rio de Janeiro, RJ 21941-972, Brazil}
\ead{mbrum@if.ufrj.br}
\address{$^{2}$ II. Institut f\"ur theoretische Physik, Universit\"at Hamburg -- Luruper Chaussee 149, D-22761 Hamburg (HH), Germany}
\ead{marcos.brum@desy.de, klaus.fredenhagen@desy.de}

\begin{abstract}
We present a modification of the recently proposed Sorkin-Johnston states for scalar free quantum fields on a class of globally hyperbolic spacetimes possessing compact Cauchy hypersurfaces. The modification relies on a smooth cutoff of the commutator function and leads
always to Hadamard states, in contrast to the original Sorkin-Johnston states. The modified Sorkin-Johnston states are, however, due to the smoothing no longer uniquely associated to the spacetime. \\[3ex]
\end{abstract}

\submitto{CQG}

\maketitle

\section{Introduction}
One of the crucial insights of quantum field theory on curved spacetimes is the absence of a distinguished state corresponding to the vacuum state on Minkowski space. This is intimately related with the nonexistence of a unique particle interpretation of the theory and manifests itself most dramatically in the Hawking effect. The absence of a vacuum state has nowadays the status of a no go theorem \cite{FewsterVerch12,FewsterVerch_sf12} which is valid under very general conditions.

Recently a new proposal for a distinguished quantum state for a free scalar field has been put forward by Sorkin and Johnston (see \cite{AfshordiAslanbeigiSorkin12} and references therein). Their idea is based on the fact that the commutator function may be considered as the integral kernel of an antisymmetric operator on some real Hilbert space, as discussed long ago e.g. by Manuceau and Verbeure \cite{ManuceauVerbeure68}. Under some technical conditions, the polar decomposition of this operator yields an operator having the properties of the imaginary unit, and a positive operator in terms of which a new real scalar product can be defined. The new scalar product then induces a pure quasifree state. This method of constructing a state can e.g. be applied for a free scalar quantum field on a static spacetime where the energy functional provides a quadratic form on the space of Cauchy  data in terms of which a Hilbert space can be defined. The result is the ground state with respect to  time translation symmetry (see, e.g. \cite{Kay78}).

On a spacetime without a timelike Killing vector it is not clear how to introduce a Hilbert space structure which is determined by the given data, the geometry and the parameters in the Klein-Gordon equation. The proposal of Sorkin and Johnston now is to use the volume measure on the spacetime and the corresponding real Hilbert space of square integrable real-valued functions \footnote{The analogous idea for the Dirac field has been proposed and analyzed some time ago by Finster \cite{Finster11}, there called the fermionic projector.}. The question which arises is whether the commutator function, considered as an antisymmetric densely defined bilinear form, admits a polar decomposition as needed for the construction of a state. Provided such a state exists one would like to see whether it satisfies the Hadamard condition which guarantees that the state can be extended to composite local fields as e.g. the energy momentum tensor. 

These questions have been investigated by Fewster and Verch \cite{FewsterVerch-SJ12}. They prove that the commutator function induces a bounded operator 
if the spacetime admits an isometric embedding as a relatively compact subset of another spacetime,
thus in this case the construction is possible. But the state, in general,  does not satisfy the Hadamard condition; moreover, its restriction to a smaller subregion may induce a GNS representation, which is inequivalent to the GNS representation induced by the S-J state of the smaller region.

While such an obstruction had to be expected in view of the mentioned no go theorem, it would be a pity if this new ansatz for the construction of states had to be abolished. As a matter of fact, our understanding of the state space of quantum field theories is still rather poor. We know, by the deformation argument of  Fulling, Narcowich and Wald \cite{FullingNarcowichWald81}, that Hadamard states on globally hyperbolic spacetimes always exist, but this argument is rather indirect and does not admit a detailed physical interpretation. On Friedmann-Robertson-Walker (FRW) spacetimes, a concrete prescription is that of adiabatic vacuum states, as introduced by Parker \cite{Parker69}. It was later made mathematically precise by L\"uders and Roberts \cite{LuRo90} and further analyzed by Junker and Schrohe \cite{JunSchrohe02}. Unfortunately, it turned out that in the precise version the prescription is no longer unique, but determines instead a class of states. Junker also gave a general construction of Hadamard states in terms of pseudo-differential operators. This method was recently generalized by G\'{e}rard and Wrochna \cite{GerardWrochna12}. Another construction applies to spacetimes with an asymptotically flat past. Here states can be interpreted by their properties on a past horizon. This is interesting for the description of states for the early universe. (See, e.g. \cite{dappiaggi:062304}.)  

Nearer to the original idea of Parker is the concept of states of low energy (SLE-states), as proposed by Olbermann \cite{Olbermann07}. Here the idea is to minimize the energy density (averaged over time) in spatially homogeneous states on FRW spacetimes. This idea is motivated by the result of Fewster that suitable averages of the energy density over a timelike curve are bounded from below (Quantum Energy Inequalities \cite{Fewster00,FewsterSmith08}). The SLE depend only on the sampling function and satisfy the Hadamard condition. Their construction was recently extended to a larger class of spacetimes \cite{ThemBrum13}. As shown by Degner \cite{Degner09}, concrete calculations based on these states are possible.

In this work we present a modification of the S-J states. 
As in \cite{FewsterVerch-SJ12} we consider an embedding of a spacetime as a relatively compact subset of another spacetime. But instead of applying the S-J-construction to the commutator function of the smaller spacetime we apply it to the commutator function of the larger spacetime, multiplied in both variables by a smooth function with compact support which is identical to 1 on the embedded spacetime. We test this approach on static and cosmological spacetimes and find that the construction yields Hadamard states. The original construction is obtained if one substitutes the smooth function by the characteristic function of the smaller spacetime.

On section \secref{fieldquant} we review the scalar field quantization according to the algebraic approach. On section \secref{S-J_states} we construct the modified S-J states, presenting the requirements imposed on the spacetime for the construction to be well-defined and showing that the smoothing is sufficient for these states to be Hadamard. 


\section{Scalar Field quantization on Globally Hyperbolic Spacetimes}\label{fieldquant}

\subsection{Quantized scalar field}

Globally hyperbolic spacetimes $\M$ are smooth, orientable, time orientable and paracompact manifolds that admit a foliation into smooth, nonintersecting spatial hypersurfaces $\Sigma$ of codimension 1 \cite{BernalSanchez03,Wald84}. They have the topological structure $\M = \mathbb{R}\times\Sigma$. For any subset $S\subset\Sigma$ one can define its {\it Domain of Dependence} $D(S)$ as the set of points $p\in \M$ such that every inextendible causal curve through $p$ intersects $S$. Clearly, $D(\Sigma)=\M$. The determination of the solution of the equations of motion on a neigborhood of $S$ fixes uniquely the field configuration at any point of spacetime contained in $D(S)$ \cite{Wald94}.

It is well known \cite{BarGinouxPfaffle07} that the Klein-Gordon equation on such a spacetime admits unique retarded and advanced fundamental solutions, which are maps $\mathds{E}^{\pm}:\mathcal{C}_{0}^{\infty}(\M)\rightarrow \mathcal{C}^{\infty}(\M)$, such that, for $f\in \mathcal{C}_{0}^{\infty}(\M)$,
\begin{equation}
 \left(\Box +m^{2}\right)\mathds{E}^{\pm}f=\mathds{E}^{\pm}\left(\Box +m^{2}\right)f=f
 \label{KGfund}
\end{equation}
and
\[\textrm{supp}(\mathds{E}^{\pm}f)\subset J^{\pm}(\textrm{supp}f) \; ,\]
where $J^{+(-)}(S)$, $S\in\M$, is the causal future (past) of $S$, the set of points on $\M$ which can be reached from $S$ along a future-(past-)directed causal curve having starting point in $S$. The functions $f\in \mathcal{C}_{0}^{\infty}(\M)$ are called test functions, and $P\coloneqq \Box +m^{2}$ will denote the differential operator. From the fundamental solutions, one defines the {\it advanced-minus-retarded-operator} $\mathds{E}\coloneqq \mathds{E}^{-}-\mathds{E}^{+}$ as a map $\mathds{E}:\mathcal{C}_{0}^{\infty}(\M)\rightarrow \mathcal{C}^{\infty}(\M)$.

The Lorentzian metric $g$ generates a measure on the spacetime, and we define the inner product on the space of test functions $\mathcal{C}_{0}^{\infty}(\M)$ by
\begin{equation}
 (f,f')_{g}\coloneqq \int\textrm{d}^{4}x\sqrt{|g|}\, \overline{f}(x)f'(x) \; .
  \label{innerprod}
\end{equation}
Using $\mathds{E}$, we define the anti-symmetric form
\begin{equation}
 \sigma(f,f')=-\int\textrm{d}^{4}x\sqrt{|g|}\, f(x)(\mathds{E}f')(x) = -(\overline{f},\mathds{E}f')_{g} \eqqcolon -E(f,f') \; .
 \label{symplform}
\end{equation}

The free quantum field $\Phi$ is a linear map from the space of test functions $\mathcal{C}_{0}^{\infty}(\M)$ to a unital *-algebra satisfying
\begin{equation}
 \Phi\left(Pf\right)=0
 \label{KGEoM}
\end{equation}
hence $\Phi$, formally written as 
\[\Phi(f)=\int\textrm{d}^{4}x\sqrt{|g|}\, \phi(x)f(x) \; ,\]
may be understood as an algebra valued distributional solution of the Klein Gordon equation.  
Moreover, $\Phi$ satisfies the relations
\begin{enumerate}[label=(\roman{*})]
 \item $\Phi(\overline{f})=\Phi(f)^{*}$;
 \item $\comm{\Phi(f)}{\Phi(f')}=-i\sigma(f,f')\mathds{1}$, where $\comm{\cdot}{\cdot}$ is the commutator and $\mathds{1}$ is the unit element.
\end{enumerate}
The $CCR$-algebra $\F$ is the (up to isomorphy) uniquely determined unital $^{*}$-algebra generated by the symbols $\Phi(f)$. 


The symbols $\Phi(f)$ are unbounded. In order to obtain an algebra of bounded operators, we will construct the so-called {\it Weyl algebra} as follows. Operating on the quocient space $\mathcal{C}_{0}^{\infty}(\M)/\mathrm{Ke}\mathds{E}_\M \times \mathcal{C}_{0}^{\infty}(\M)/\mathrm{Ke}\mathds{E}_\M$, the anti-symmetric form $\sigma$ becomes nondegenerate. We thus define the real vector space $L\coloneqq \mathrm{Re}\left(\mathcal{C}_{0}^{\infty}(\M)/\mathrm{Ke}\mathds{E}_\M\right)$ and hence $(L,\sigma)$ is a real symplectic space where $\sigma$ is the symplectic form. A $C^{*}$-algebra can be formed from the elements of this real symplectic space by introducing the symbols $W(f)$, $f\in L$ which satisfy (for more details, see\cite{BraRob-II}):
\begin{enumerate}[label=(\Roman{*})]
 \item $W(0)=\mathds{1}$;
 \item $W(-f)=W(f)^{*}$;
 \item For $f,g\in L$, $W(f)W(g)=e^{-i\frac{\sigma(f,g)}{2}}W(f+g)$.
\end{enumerate}
These symbols generate the {\it Weyl algebra} $\W\left(L,\sigma\right)$. From the nondegenerateness of the symplectic form one sees that $W(f)=W(g)$ iff $f=g$.

In the following we want to restrict ourselves to two classes of spacetimes for which we have good control on the commutator function $\mathds{E}$: the first class consists of static spacetimes, i.e. spacetimes with a timelike Killing vector $k$ and with Cauchy surfaces which are orthogonal to the Killing vector; these spacetimes admit a coordinate system in which the metric assumes the form 
\begin{equation}
ds^2=\alpha^{2}dt^2-h_{ij}dx^idx^j
\label{metric-timedecomp-static}
\end{equation}  
where all coefficients are smooth functions on a Cauchy surface $\Sigma$.
The  second class, called expanding spacetimes, has a metric of the form
\begin{equation}
 ds^{2}=dt^{2}-c^{2}h_{ij}dx^{i}dx^{j}\; ,
 \label{metric-timedecomp-GH}
\end{equation}
where $c$ is a smooth positive function of time, the so-called {\it scale factor}, and $h_{ij}$ is a  smooth time independent Riemannian metric on $\Sigma$. This class includes in particular cosmological spacetimes of the Friedmann-Robertson-Walker type. The ultrastatic spacetimes investigated in \cite{FewsterVerch-SJ12} belong to both classes. For simplicity, we will only consider spacetimes with compact Cauchy hypersurfaces.

Actually, there always exists a coordinate system in which the metric on a globally hyperbolic spacetime assumes the form \cite{Wald84}, 
\begin{equation*}
 ds^{2}=\gamma(t,\uline{x}) dt^{2}-h_{ij}(t,\uline{x})dx^{i}dx^{j}\; .
\end{equation*}
We expect that, with some more effort, our constructions can be generalized to the  generic case.

In static spacetimes, the Klein-Gordon equation \eqref{KGfund} becomes
\begin{equation}
 \frac{\partial^{2}\phi}{\partial t^{2}}+K\phi=0 \; ,
  \label{KG-static}
\end{equation}
where
\[K=\alpha^{2}\left[\frac{1}{\sqrt{|g|}}\partial_{j}(\sqrt{|g|}h^{jk}\partial_{k})+m^{2}\right] \; .\]
On the Hilbert space $L^2(\Sigma,\alpha^{-2}\sqrt{|g|})$, the operator $K$ is symmetric and positive. According to Kay \cite{Kay78}, if the spacetime is uniformly static, i.e., $\alpha$ is bounded fom above and from below away from zero, the operator $K$ is even essentially selfadjoint on the domain $\mathcal{C}^{\infty}_0(\Sigma)$. We give in Theorem \secref{theorem_essential_self-adjointness} a proof of essential self-adjointness of $K$ without any assumptions on $\alpha$.
Due to the compactness of $\Sigma$, its selfadjoint closure, again denoted by $K$, has a discrete spectrum with an orthonormal system of smooth eigenfunctions $\psi_{j}$ and positive eigenvalues $\lambda_j$, $j\in\mathbb{N}$ with $\lambda_j\ge\lambda_k$ for $j>k$. Moreover, due to Weyl's asymptotic, the sums $\sum_j\lambda_j^{-p}$ converge for $p>\frac{d}{2}$ where $d$ is the dimension of $\Sigma$ \cite{Jost11}.

The advanced-minus-retarded-operator, in this case, has the integral kernel
\begin{equation}
\mathds{E}(t,\uline{x};t',\uline{x'})=-\sum_{j}\frac{1}{\omega_{j}}\sin((t-t')\omega_{j})\psi_{j}(\uline{x})\overline{\psi}_{j}(\uline{x'})\; ,
\label{E-static}
\end{equation}
with $\omega_{j}=\sqrt{\lambda_j}$ \cite{Fulling89}. This sum converges in the sense of distributions, as the one in equation \eqref{E-prop}.

In the case of expanding spacetimes, the Klein-Gordon equation assumes the form (see, e.g., \cite{LuRo90,ThemBrum13})
\begin{equation}
 \left(\partial_{t}^{2}+3\frac{\dot{c}(t)}{c(t)}\partial_{t}-\frac{\Delta_{h}}{c(t)^{2}}+m^{2}\right)\phi(t,\uline{x})=0 \; .
 \label{KG}
\end{equation}

On the compact Riemannian space $(\Sigma,h)$ the Laplace operator $-\Delta_{h}$ is essentially self-adjoint. Its unique self-adjoint extension (denoted by the same symbol) is an operator on $L^{2}(\Sigma,\sqrt{\lvert h\rvert})$ with discrete spectrum \cite{Jost11}.
Again we use the orthonormal basis of eigenfunctions $\psi_{j}$ and the associated nondecreasing sequence of eigenvalues $\lambda_j$ of $-\Delta_h$.
An ansatz for a solution is 
\begin{equation}
\Phi(t,\uline{x})=T_j(t)\psi_{j}(\uline{x}) \ .
\end{equation}
$T_j$ then has to satisfy the ordinary second order linear differentlal equation
\begin{equation}
\frac{d^2}{dt^2}T_j+3\frac{\dot{c}}{c}\frac{d}{dt}T_j+\omega_j^2T_j=0
\end{equation}
with
\begin{equation}
 \omega_{j}(t)\coloneqq\sqrt{\frac{\lambda_{j}}{c(t)^{2}}+m^{2}} \; .
 \label{KGfreq}
\end{equation}
The 2 linearly independent real-valued solutions of this equation can be combined into  one complex valued solution satisfying the normalization condition
\begin{equation}
T_j(t)\dot{\overline{T}}_j(t)-\dot{T}_j(t)\overline{T}_j(t)=\frac{i}{c(t)^3}\ .
\end{equation} 

The advanced-minus-retarded operator now has the integral kernel
\begin{equation}
 \mathds{E}(t,\uline{x};t',\uline{x'})=\sum_{j}\frac{(\overline{T}_{j}(t)T_{j}(t')-T_{j}(t)\overline{T}_{j}(t'))}{2i}\psi_{j}(\uline{x})\overline{\psi}_{j}(\uline{x'}) \; .
 \label{E-prop}
\end{equation}

\subsection{States and the Hadamard condition}\label{secstates}$ $

States $\omega$ are functionals over the algebra $\F(\M)$ (or $W(L)$), with the following properties: 
\begin{description}
 \item [Linearity] $\omega(\alpha A+\beta B)=\alpha\omega(A)+\beta\omega(B)$, $\alpha$, $\beta\in \mathbb{C}$, $A$, $B\in \F(\M)$ (or $A$, $B\in W(L)$);
 \item [Positive-semidefiniteness] $\omega(A^{*}A)\geq 0$;
 \item [Normalization] $\omega(\mathds{1})=1$.
\end{description}
The $n-${\it point functions} of $\omega$ are defined as
\[w_{\omega}^{(n)}(f_{1}\otimes \ldots \otimes f_{n}) \coloneqq \omega(\Phi(f_{1})\ldots\Phi(f_{n}))\]
(or the corresponding relation for the Weyl algebra). In the present work we will focus on states which are completely described by their two-point function, the so called {\it Quasifree States}. The two-point function can be decomposed in its symmetric and anti-symmetric parts ($f_{1},\, f_{2}\in L$ below)
\[w_{\omega}^{(2)}(f_{1},f_{2})=\mu(f_{1},f_{2})+\frac{i}{2}\sigma(f_{1},f_{2}) \; ,\]
where $\mu(\cdot,\cdot)$ is a symmetric product which majorizes the symplectic product, i.e.
\[|\sigma(f_{1},f_{2})|^{2}\leq 4\mu(f_{1},f_{1})\mu(f_{2},f_{2}) \; .\]
The state is said to be {\it pure} if the inequality above is saturated, i.e., $\forall f_{1}\in L$ and $\forall
\epsilon>0$, $\exists \, f_{2}\in L$ such that
\[\frac{|\sigma(f_{1},f_{2})|^{2}}{\mu(f_{2},f_{2})}\ge (4-\epsilon)\mu(f_{1},f_{1}) \; .\]


In order to extend the states to correlation functions of nonlinear functions of the field as, e.g., the energy momentum tensor, one needs some control on the singularities of the $n$-point functions. On Minkowski space, the spectrum condition implies such a structure, and the standard way to incorporate nonlinear functions of the field is via normal ordering. On a generic spacetime one replaces the spectrum condition by a condition on the wavefront set. As shown by Radzikowski \cite{Radzikowski96}, the two-point functions of Hadamard states can elegantly be characterized by their wave front sets. This observation is at the basis of the modern approach to quantum field theory on curved spacetimes   \cite{BruFreKoe96,BruFre00,HoWa02}.

The wavefront set of a distribution $v$ is a subset of the cotangent bundle 
which characterizes its singularity. Roughly speaking, it consists of elements $(x,k)\in \mathcal{T}^*\M$, $k\not=0$ such that the local Fourier transform of $v$ does not decay
rapidly in any open cone $V$ around $k$.
Rapid decay means that $\forall N \in \mathbb{N}_{0} \; ,\, \exists C_{N}>0$ such that
\begin{equation}
 \lvert \hat{v}(k) \rvert \leqslant C_{N}\left(1+\lvert k \rvert\right)^{-N} \; ,\, k \in V \; ,
 \label{regsupp}
\end{equation}
where $\hat{v}$ is a local Fourier transform of $v$ at $x$, i.e. the Fourier transform (in any chart of $\M$) of $\phi v$ for a test function $\phi$ with compact support which does not vanish at $x$.

Finally, Hadamard states are defined by the following
\begin{mydef}
 A state $\omega$ is said to be a {\it Hadamard state} if its two-point distribution $\omega_{2}$ has the following wavefront set:
\begin{equation}
 WF(\omega_{2})=\left\{\left(x_{1},k_{1};x_{2},-k_{2}\right) | \left(x_{1},k_{1};x_{2},k_{2}\right)\in {\mathcal T}^{*}\left(\M\times\M\right) \diagdown \{0\} ; (x_{1},k_{1})\sim (x_{2},k_{2}) ; k_{1}\in \overline{V}_{+}\right\}
 \label{Wfcond}
\end{equation}
where $(x_{1},k_{1})\sim (x_{2},k_{2})$ means that there exists a null geodesic connecting $x_{1}$ and $x_{2}$, $k_{1}$ is the cotangent vector to this geodesic at $x_{1}$ and $k_{2}$, its parallel transport, along this geodesic, at $x_{2}$. $\overline{V}_{+}$ is the closed forward light cone of $\mathcal{T}^{*}_{x_{1}}\M$.
\end{mydef}

Since the antisymmetric part of a Hadamard 2 point function is the commutator function $\mathds{E}$, 
the difference between the two-point functions of different Hadamard states is symmetric. But the symmetric part of the wave front set of a Hadamard function is empty, hence it is a smooth function. 
This fact will play a fundamental role when we come to the proof that the states which will be constructed below are Hadamard states.

Later we will also need a refinement of the concept of the (smooth) wavefront set, namely the Sobolev wavefront set of order $s$ with $s\in\mathbb{R}$. It is obtained from the definition above by the replacement of the condition of rapid decay within a cone $V$ by the condition
\begin{equation}
\int_V d^n k\,(1+|k|^2)^s|\hat{v}(k)|^2<\infty \ .
\end{equation}
For the complete definition of wavefront sets, see \cite{Hormander-I}. 

\section{``Vacuum-like'' Hadamard states}\label{S-J_states}

As stated in the introduction, the original definition of the Sorkin-Johnston states aimed at constructing distinguished states on any globally hyperbolic spacetime \cite{AfshordiAslanbeigiSorkin12}. This was supposed to fill the gap left open by the absence of a vacuum state on nonstationary spacetimes, as well as serving as initial state for application in cosmological problems. Actually, on Minkowski space one could show that it indeed coincides with the vacuum (modulo some technical problems with unbounded bilinear forms).

Unfortunately, it turned out that in typical cases which are under control the resulting states are not Hadamard states \cite{FewsterVerch-SJ12}. We are going to present now a modification of this construction, that we call modified S-J states. After presenting the general construction, we show that we obtain Hadamard states on both static and expanding spacetimes.

The construction of the S-J states starts from the observation in \cite{FewsterVerch-SJ12} that the advanced-minus-retarded-operator, operating on square-integrable functions on a globally hyperbolic spacetime $\M$, embedded, with relatively compact image, into another globally hyperbolic spacetime $\N$, is a bounded operator.

We consider a globally hyperbolic spacetime $\N=\mathbb{R}\times\Sigma$ with compact Cauchy surfaces $\{t\}\times\Sigma$ and a subspacetime $\M=\mathds{I}\times\Sigma$, where $\mathds{I}=(a,b)$ is a bounded interval. We have the isometric embedding $\Psi :\M \rightarrow \N,\ (t,\uline{x})\mapsto (t,\uline{x})$.

By the uniqueness of the advanced and retarded fundamental solutions the  advanced-minus-retarded-operator on $\M$ is obtained from the corresponding operator on $\N$,
\begin{equation}
\mathds{E}_{\M}=\Psi^{*}\mathds{E}_{\N}\Psi_{*}\; ,
\end{equation}
where $\Psi^{*}$, $\Psi_{*}$ are, respectively, the pull-back and push-forward associated to $\Psi$. $\Psi_{*}$ is an isometry from $L^2(\M)$ to $L^2(\N)$, and $\Psi^{*}$ its adjoint.

Let $f\in \mathcal{C}_{0}^{\infty}(\N)$ be a real-valued test function such that $f\equiv 1 $ on $\Psi(\M)$. We define the bounded self-adjoint operator $A$
\begin{equation}
 {A} \coloneqq if\mathds{E}_\N f \; ,
 \label{A-iE}
\end{equation}
where $f$ acts by multiplication on $L^2(\N)$. If we replace $f$ by the characteristic function of $\M$ we obtain the operator analyzed in \cite{FewsterVerch-SJ12}.

A state can then be constructed in the same way as in the quoted literature by taking the positive part $A^+$ of $A$ (in the sense of spectral calculus).
\begin{equation}
 A^{+}=P^+A,
\end{equation}
where $P^+$ is the  spectral projection on the interval $[0,||A||]$.

The modified S-J state $\omega_{SJ_{f}}$ is now defined as the quasifree state on the spacetime $\M$ whose two-point function is given by
\begin{equation}
W_{SJ_{f}}(q,r) \coloneqq (  q , A^{+}r  ) \; ,
\label{W-SJ}
\end{equation}
for real-valued test functions $q,r$ on $\M$. Note that the antisymmetric part of  the two-point function coincides with $i\mathds{E}_\M$. This is due to the fact that the intersection of the kernel of $A$ with $L^2(\M)$ coincides with the kernel of $\mathds{E}_\M$. In particular, the integral kernel of $A^{+}$, restricted to $\M$, is a bisolution of the Klein-Gordon equation. This bisolution can be uniquely extended to the domain of dependence of $\M$ (which coincides with $\N$ in the case considered here). The state $\omega_{SJ_{f}}$ is a pure state, as can be seen in the following

\begin{Theorem}
Let $\M$ be a globally hyperbolic subspacetime of another globally hyperbolic spacetime $\N$, and let $\Sigma\subset\M$ be a Cauchy surface of $\N$. Then for every real-valued $f\in\mathcal{C}^{\infty}_0(\N)$ with $f\equiv1$ on $\M$, the modified S-J state
\[\omega_{SJ_f}(W(\phi))=e^{-\frac12(\phi,|f\mathds{E}_\N f|\phi)}\]
with $\phi\in\mathcal{C}^{\infty}_0(\M)$, is pure. Here $\mathds{E}_\N$ is the commutator function on $\N$ and $|\cdot|$ denotes the modulus of the operator.
\end{Theorem}
\begin{proof}
We consider the Weyl algebra over the symplectic space $(L,\sigma)$ with, now, 
\[L=\mathrm{Re}\left(\mathcal{C}^{\infty}_0(\N)/\mathrm{Ke}\mathds{E}_\N f\right)\]
and
\[\sigma([\phi_1],[\phi_2])=(\phi_1,f\mathds{E}_\N f\phi_2) \; .\]
Due to the compactness of the support of $f$, the operator $f\mathds{E}_\N f$ is bounded on this Hilbert space, and, according to the results of Manuceau and Verbeure \cite{ManuceauVerbeure68} mentioned in the Introduction, we can define a pure state on the Weyl algebra by setting 
\[\omega(W(\phi))=e^{-\frac12(\phi,|f\mathds{E}_\N f|\phi)}\]
where $|f\mathds{E}_\N f|=\sqrt{-f\mathds{E}_\N f^2\mathds{E}_\N f}$. 

It remains to prove that the Weyl algebra above coincides with the Weyl algebra over $\M$ with the symplectic form defined by the commutator function $\mathds{E}_\M$ on $\M$. 
For this purpose we prove that the corresponding symplectic spaces are equal. Since the restriction of $\mathds{E}_\N$ to $\M$ coincides with $\mathds{E}_\M$ and since $f\equiv1$ on $\M$, the symplectic space associated to $\M$ is a symplectic subspace of $(L,\sigma)$. We now show that this subspace is actually equal to $(L,\sigma)$. This amounts to prove that every rest class $[\phi]\in L$ with $\phi\in\mathcal{C}^{\infty}_0(\N)$ contains an element $\phi_0$ with $\mathrm{supp}\phi_0\subset \M$.

Here we proceed similarly to Fulling, Sweeny and Wald \cite{FullingSweenyWald78}. Let $\phi\in\mathcal{C}^{\infty}_0(\N)$. We may decompose $\phi=\phi_++\psi+\phi_-$ with $\mathrm{supp}\phi_\pm\subset J_\pm(\Sigma)$ and $\mathrm{supp}\psi\subset\M$. Let $\chi\in\mathcal{C}^{\infty}(\N)$ such that $\chi\equiv1$ on $J_+(\Sigma)$ and $\mathrm{supp}\chi\subset J_+(\Sigma_-)$ for a Cauchy surface $\Sigma_-$ of $\M$ in the past of $\Sigma$. Set
\[\psi_+=P(1-\chi)\mathds{E}^-_\N f\phi_+\]
where $P$ is the Klein Gordon operator and $\mathds{E}_\N^-$ the advanced propagator. By the required properties of $\chi$, $\psi_{+}$ vanishes where $\chi$ is constant, hence $\mathrm{supp}\psi_+\subset \M$. In particular $f\psi_+=\psi_+$. We are left with showing that $\phi_+-\psi_+\in\mathrm{Ke}E_\N f$,
\[\mathds{E}_\N f(\phi_+-\psi_+)=\mathds{E}_\N(f\phi_+-\psi_+)=\mathds{E}_\N P\chi \mathds{E}^-_\N f\phi_+=0\ ,\]
where in the last step we used the fact that $\chi \mathds{E}^-_N f\phi_+$ has compact support. For $\phi_-$ an analogous argument works and yields an element $\psi_-\in[\phi_-]$ with $\mathrm{supp}\psi_-\subset \M$. Thus we find that $\phi_0=\psi_++\psi+\psi_-$ has the properties required above.

\end{proof}

The question now arises whether the modified S-J states are Hadamard states. We will prove this to be true in two situations, static spacetimes and expanding spacetimes. We remark that the proofs rely only upon the fact that $f\in \mathcal{C}_{0}^{\infty}(\N)$ is a real-valued test function such that $f_{\upharpoonright \M}\equiv 1 $. 
If we change $f$ we will obtain in general a different Hadamard state. Thus the states we construct here are not uniquely singled out by the spacetime geometry.

In both types of spacetime, the operator $\mathds{E}$ can be decomposed into a sum over the eigen projections $|\psi_{j}\rangle\langle \psi_{j}|$ of the spatial part of the Klein-Gordon operator. We choose our cutoff function $f$ to depend only on time. It remains then to analyze for each $j$ the operators $A_j$ defined as
\begin{equation}
A(t,\uline{x};t',\uline{x'}) \eqqcolon \sum_{j}A_{j}(t',t)\psi_{j}(\uline{x})\overline{\psi}_{j}(\uline{x'})\; .
\label{A_j}
\end{equation}

\subsection{Static spacetimes}

Taken as an operator on $L^{2}(\mathbb{R})$, $A_{j}$ has the integral kernel (see \eqref{E-static})
\begin{equation}
 A_{j}(t',t)=\frac{i}{\omega_{j}}f(t')\left(\sin (\omega_{j}t'-\theta_j)\cos (\omega_{j}t-\theta_j)-\cos (\omega_{j}t'-\theta_j)\sin (\omega_{j}t-\theta_j)\right)f(t).
\end{equation}
This expression does not depend on the phase $\theta_j$ due to the addition theorem of trigonometric functions.
We choose 
$\theta_j$ such that
\[\int dtf(t)^2\cos(\omega_jt-\theta_j)\sin(\omega_jt-\theta_j)=0\ .\]
Such a choice is possible since the integrand changes its sign if $\theta_j$ is shifted by $\pi/2$.

Since $A_{j}^{*}(t',t)\equiv \overline{A_{j}(t,t')}$, we find
\begin{equation}
|A_j|(t',t)\equiv \left(A_{j}^{*}A_{j}\right)^{1/2}(t',t)=\frac{1}{\omega_j^2}\left(||S_j||^2C_j(t)C_j(t')+||C_j||^2S_j(t)S_j(t')\right)
\end{equation}  
with
\[S_j(t)=f(t)\sin(\omega_jt-\theta_j)\ ,\ C_j(t)=f(t)\cos(\omega_jt-\theta_j)\ ,\]
\[||S_j||^2 \coloneqq \int dtS_{j}(t)^{2}\; ,\]
and similarly for $||C_j||^2$. We further note that $A_{j}^{+}=(A_{j}+|A_{j}|)/2$. Hence the positive part of $A_j$ has the integral kernel
\[A_j^+(t',t)=\frac{1}{2\omega_j||C_j||||S_j||}\left(||S_j||C_j(t)-i||C_j||S_j(t)\right)\left(||S_j||C_j(t')+i||C_j||S_j(t')\right)\ .\]

Setting
\begin{equation}
 \delta_{j}\coloneqq 1-\frac{||C_{j}||}{||S_{j}||}\; ,
 \label{delta}
\end{equation}
we write
\begin{equation}
 A_{j}^{+}(t,t')=
 \frac{1}{2\omega_j}\left(\frac{1}{1-\delta_j}C_j(t)-iS_j(t)\right)\left(C_j(t')+i(1-\delta_j)S_j(t')\right)\; .
\end{equation}
Therefore, the arising two-point function on $\M$ is
\begin{align}
 W_{SJ_{f}}(t,\uline{x};t',\uline{x}')=\sum_{j}\frac{1}{2\omega_{j}}&\left(\frac{1}{1-\delta_{j}}C_j(t)-iS_j(t)\right)\left(C_j(t')+i(1-\delta_j)S_j(t')\right)\psi_j(\uline{x}) \overline{\psi}_j(\uline{x}') \; .
 \label{W-SJ-static}
\end{align}

A practical way to verify that this state is a Hadamard state is to compare it with another Hadamard state and check whether the difference $w$ of the two-point functions is smooth. For this comparison, we use the two-point function of the static ground state, restricted to $\M$. 
\begin{equation}
 W_{0}(t,\uline{x};t',\uline{x}')=\sum_{j}\, \frac{e^{-i\omega_{j}(t-t')}}{2\omega_{j}}\psi_{j}(\uline{x})\overline{\psi}_{j}(\uline{x}') \; .
 \label{W-H-static}
\end{equation}
For $\delta_j=0$ it coincides with \eqref{W-SJ-static}. Further we note that multiplying this function by $f(t)f(t')$ gives the same function, since $f_{\upharpoonright \M}\equiv 1$.


We state our result as a theorem:

\begin{Theorem}\label{Had-static}
Let $\N=\mathbb{R}\times \Sigma$ be a static spacetime with metric $g=a^2dt^2-h$, where $h$ is a Riemannian metric on the compact manifold $\Sigma$ and $a$ is a smooth everywhere positive function on $\Sigma$. Let $I$ be a finite interval and $f$ a smooth real-valued function on $\mathbb{R}$ with compact support which is identical to 1 on $I$. Then the modified S-J-state $\omega_{SJ_f}$ as constructed above on $\M=I\times\Sigma$ is a Hadamard state.  
\end{Theorem}

\begin{proof}
The difference $\colon W_{SJ_{f}}\colon$ between $ W_{SJ_{f}}$ and $W_{0}$ is 
\begin{equation}
 \colon W_{SJ_{f}}\colon (t,\uline{x};t',\uline{x}')=\sum_{j}\, \frac{\delta_{j}}{2\omega_{j}}\left[\frac{1}{1-\delta_{j}}C_{j}(t')C_{j}(t)-S_{j}(t')S_{j}(t)\right]\psi_{j}(\uline{x})\overline{\psi}_{j}(\uline{x}') \; .
 \label{W-SJH-static}
\end{equation}
To prove that $\omega_{SJ_{f}}$ is a Hadamard state it suffices to show that $\colon W_{SJ_{f}}\colon$ is smooth.
Since the eigenfunctions $\psi_j$ of the elliptic operator $K$ are smooth, each term in the expansion above is smooth, and it suffices to prove that the sum 
converges in the sense of smooth functions. This can be done by proving that, for all derivatives, the sum converges in $L^2(\M\times\M)$.

For this purpose we first exploit that the $L^2$-norms of derivatives of functions on $\Sigma$ can be estimated in terms of the operator $K$. Namely, for every differential operator $D$ of order $n$ on $\Sigma$ there exists a constant $c_D>0$ such that 
\[||D\psi||_2\le c_D||K^{m}\psi||_2\]  
with $m$ the smallest integer larger than or equal to $n/2$ \cite{Hormander-III}. Hence spatial derivatives of the functions $\psi_j$ can be absorbed by multiplication with the corresponding eigenvalues of $K$. Similarly, time derivatives amount to multiplication with factors $\omega_j$ and exchanges between the functions $S_j$ and $C_j$. Since their $L^2$-norms are uniformly bounded in $j$, it remains to show that
\[\sum_j\omega_j^n\delta_j<\infty\ \forall\ n\in\mathbb{N}_0 \ .\]
We first observe that $||C_{j}||^{2}$ and $||S_{j}||^{2}$ can be expressed in terms of the Fourier transform of the square of the test function $f$:
\begin{equation}
||C_{j}||^{2}=\int\textrm{d}t\, f(t)^{2}\left(\frac{e^{2i(\omega_{j}t-\theta)}+e^{-2i(\omega_{j}t-\theta)}+2}{4}\right)=\frac{1}{2}+\frac{\widetilde{f^{2}}(2\omega_{j})e^{-2i\theta}+\widetilde{f^{2}}(-2\omega_{j})e^{2i\theta)}}{4}
\end{equation}
and
\begin{equation}
||S_{j}||^{2}=\frac{1}{2}-\frac{\widetilde{f^{2}}(2\omega_{j})e^{-2i\theta}+\widetilde{f^{2}}(-2\omega_{j})e^{2i\theta}}{4} \; .
\end{equation}
Since $f$ is a smooth test function, so is $f^{2}$, and $\forall n\in\mathds{R}$,
\begin{equation}
\lim_{\omega\rightarrow \infty}\omega^{n}\widetilde{f^{2}}(2\omega)=0 \; .
\label{omega-f-smooth}
\end{equation}
It follows immediately that
\begin{equation}
\lim_{j\rightarrow \infty}\omega_{j}^{n}\delta_{j}=0 \; .
\label{omega-delta-smooth}
\end{equation}
 
The last information we need concerns the behavior of the eigenvalues of $K$. Here we use the fact that an elliptic operator on a $d$-dimensional compact space has a resolvent which is in the Schatten classes $L_{d/2+\epsilon}$, $\epsilon>0$ \cite{LotoRohl11}. Hence $\sum_j\omega_j^{-p}<\infty$ for a suitable $p\in\mathbb{N}$, and we finally obtain the estimate
\[\sum_j\omega_j^n\delta_j\le (\sum_j\omega_j^{-p})(\sup_k\omega_k^{n+p}\delta_k)\le \infty\ .\]
\end{proof}

Before we proceed to the case of expanding spacetimes, we remark that the smoothness of the function $f$ was crucial for getting a Hadamard state. The state depends via the expansion coefficients $\delta_j$ and the phases $\theta_j$ on the values of the Fourier transform of $f^2$ at the points $2\omega_j$, and it is the fast decrease of these values as $j$ tends to infinity which implies the Hadamard property. Hence,
if $f\notin \mathcal{C}_{0}^{\infty}(\mathbb{R})$ then, in general,  \eqref{omega-f-smooth} and \eqref{omega-delta-smooth} would not be satisfied, and the state would not be a Hadamard state.

\subsection{Expanding spacetimes}

The advanced-minus-retarded-operator is now
\begin{equation}
 \mathds{E}(t,\uline{x};t',\uline{x'})=\sum_{j}\frac{(\overline{T}_{j}(t)T_{j}(t')-T_{j}(t)\overline{T}_{j}(t'))}{2i}\psi_{j}(\uline{x})\overline{\psi}_{j}(\uline{x'}) \; .
 \label{E-prop-expanding-st}
\end{equation}
We decompose $fT_j$ into its real and imaginary parts, $fT_j=B_{j}-iD_j$, and obtain for the integral kernel of the operator $A_j$
\begin{equation}
 A_j(t',t)=i\left(D_{j}(t')B_{j}(t)-B_{j}(t')D_{j}(t)\right) \; .
\end{equation}
$A_j$ is a self-adjoint antisymmetric rank 2 operator.

We can choose the phase of $T_{j}$ such that
\begin{equation}
 \int B_{j}(t)D_{j}(t)dt \equiv 0 \; .
 \label{B-D-0}
\end{equation}

Analogous to the static case we obtain
\begin{equation}
 A_{j}^{+}(t',t)=
 \frac{1}{2||B_{j}||||D_{j}||}\left(||D_{j}||B_{j}(t')-i||B_{j}||D_{j}(t')\right)\left(||D_{j}||B_{j}(t)+i||B_{j}||D_{j}(t)\right) \; .
 \label{A-pos}
\end{equation}
Setting again
\begin{equation}
\delta_j=1-\frac{||B_{j}||}{||D_{j}||} \; ,
\end{equation}
we find for the two-point function of the modified S-J state on $\M$
\begin{equation}
 W_{SJ_{f}}(t,\uline{x};t',\uline{x}')=\sum_j \frac{1}{2}\left(\frac{1}{1-\delta_j}B_{j}(t')-iD_{j}(t')\right)\left(B_{j}(t)+i(1-\delta_j)D_{j}(t)\right)\psi_{j}(\uline{x})\overline{\psi}_{j}(\uline{x}') \; .
 \label{W-SJ-expand}
\end{equation}

We now investigate the wavefront set of this two-point function. We proceed as in the proof of the Hadamard condition for states of low energy \cite{Olbermann07,ThemBrum13} by comparing it with the two-point functions of adiabatic states of finite order. According to \cite{JunSchrohe02} adiabatic states of order $n$ have the same Sobolev wavefront sets as Hadamard states if $s<n+\frac{3}{2}$. It therefore suffices to prove that for all $n$ the two-point functions \eqref{W-SJ-expand} and that corresponding to an adiabatic state differ only by a function which is in the local Sobolev space of order $s$ satisfying the above inequality. We will present some further properties of adiabatic states before proceeding to the proof that $\omega_{SJ_{f}}$ is a Hadamard state.

We choose for the solution $T_{j}$ the solution with the initial conditions at $t_{0}$ implied by the $n$-fold iteration of the adiabatic ansatz. For sufficiently large $j$, $T_{j}$ is uniquely determined. It can be approximated by the WKB form
\begin{equation}
W_{j}^{(n)}(t)=\frac{1}{\sqrt{2\Omega_{j}^{(n)}c(t)^{3}}}\exp\left(i\int_{t_{0}}^{t}dt'\Omega_{j}^{(n)}(t')\right) \; .
\end{equation} 
Here $\Omega_{j}^{(n)}$ is recursively determined from
\begin{align}
 \Omega_{j}^{(0)} &=\omega_{j} \nonumber \\
 (\Omega_{j}^{(n+1)})^{2} &=\omega_{j}^{2}-\frac{3(\dot{c})^{2}}{4c^{2}}-\frac{3\ddot{c}}{2c}+\frac{3(\dot{\Omega}_{j}^{(n)})^{2}}{4(\Omega_{j}^{(n)})^{2}}-\frac{\ddot{\Omega}_{j}^{(n)}}{2\Omega_{j}^{(n)}} \; .
 \label{WKBiteration}
\end{align}
The authors of \cite{LuRo90} proved that for each $n$ there exists some $\lambda>0$ such that for $\lambda_j>\lambda$ the $n$-fold recursion above is well defined ($-\lambda_{j}$ are the eigenvalues of the Laplace operator - see equations \eqref{KG}-\eqref{KGfreq}). Furthermore, they proved that $\Omega_j^{(n)}$ is bounded from below by a constant times $\sqrt{\lambda_j}$, and together with its derivatives, bounded from above by constants times  $\sqrt{\lambda_j}$.
  
The solution $T_j$, at a generic time $t$, can be written as
\begin{equation}
 T_{j}(t)=\left( \alpha_{j}^{(n)}(t)W_{j}^{(n)}(t)+\beta_{j}^{(n)}(t)\overline{W}_{j}^{(n)}(t) \right)e^{i\theta_{j}} \; ,
  \label{adiabatic-lambda}
\end{equation}
where $\theta_{j}$ is the phase factor introduced so that \eqref{B-D-0} is satisfied, and the functions $\alpha_j^{(n)}$ and $\beta_j^{(n)}$ satisfy the estimates (uniformly in $t$ within a bounded interval)
\begin{align}
 \lvert 1-\alpha_{j}^{(n)}(t) \rvert &\leqslant C_{\alpha}(1+\lambda_{j})^{-n-1/2} \nonumber \\
 \lvert \beta_{j}^{(n)}(t) \rvert &\leqslant C_{\beta}(1+\lambda_{j})^{-n-1/2} \; .
\label{alpha-beta-lambda}
\end{align}

The proof that the two-point function \eqref{W-SJ-expand} has the Hadamard property is presented in the following

\begin{Theorem}
Let $\N=J\times \Sigma$ be an expanding spacetime with $\Sigma$ compact and $J$ an open interval on the real axis. Let $I$ be a finite open interval with closure contained in $J$, and let 
$f\in\mathcal{C}^{\infty}_0(J)$ such that $f$ is equal to 1 on $I$. Then   
the modified Sorkin-Johnston state $\omega_{SJ_{f}}$ as defined above is a Hadamard state on the expanding spacetime $\M=I\times\Sigma$.
\end{Theorem}

\begin{proof}
We want to show that for each $s>0$ there is an $n\in\mathbb{N}$ such that the difference of the 2-point functions of the state $\omega_{SJ_f}$ and the adiabatic state of $n$th order is an element of the Sobolev space of order $s$. As in the static case we use the fact that spatial derivatives can be estimated in terms of the elliptic operator and amount to multiplication with powers of the corresponding eigenvalues 
$\lambda_j$. For the time derivatives we exploit that the functions $T_j$ are solutions of a second order differential equation which again allows to replace derivatives by multiplication with powers of $\lambda_j$.  

Therefore, In order to verify the Hadamard property of $W_{SJ_{f}}$, we investigate for which index $s\in\mathbb{R}$ the operator
\begin{equation}
R_s=\sum_j \lambda_j^{s}(A^+_j-\frac{1}{2}|fT_j\rangle\langle fT_j|)\otimes |\psi_{j}\rangle\langle \psi_{j}|
\end{equation}  
is Hilbert-Schmidt. For this purpose we have to estimate the $L^{2}$ scalar products of the WKB functions. We have
\begin{equation}
\left(fW_j^{(n)},fW_j^{(n)}\right)_{L^{2}}=\int dtf(t)^2\frac{1}{2c(t)\Omega_j^{(n)}(t)}
\label{fW.fW}
\end{equation} 
which can be bounded from above and from below by a constant times $(1+\lambda_j)^{-\frac12}$.
On the other hand, the scalar product
\begin{equation}
\left(\overline{fW}_j^{(n)},fW_j^{(n)}\right)_{L^{2}}=\int dtf(t)^{2}\frac{1}{2c(t)\Omega_j^{(n)}(t)}\exp{2i\int_{t_0}^t\Omega_j^{(n)}(t')dt'}
\label{fW*.fW}
\end{equation}
is rapidly decaying in $\lambda_j$. This follows from the stationary phase approximation. It can be directly seen by exploiting the identity
\[\exp{2i\int_{t_0}^t\Omega_j^{(n)}(t')dt'}=\frac{1}{2i\Omega_j^{(n)}(t)}\frac{\partial}{\partial t}\exp{2i\int_{t_0}^t\Omega_j^{(n)}(t')dt'}\]
several times and subsequent partial integration. The estimates on $\Omega_j^{(n)}$ and its derivatives together with the smoothness of $c$ and $f$ then imply the claim.

Now, the term $(A^+_j-\frac{1}{2}|fT_j\rangle\langle fT_j|)$ reads
\begin{align}
 A_{j}^{+}(t',t)-\frac{f(t')T_{j}(t')f(t)\overline{T}_{j}(t)}{2} &=\frac{1}{8(1-\delta_{j})}f(t')\left\{(\delta_{j})^{2}\left(\overline{T}_{j}(t')T_{j}(t)+T_{j}(t')\overline{T}_{j}(t)\right) \right. \nonumber \\
 &\left. +2\textrm{Re}\left[\delta_{j}(2-\delta_{j})T_{j}(t')T_{j}(t)\right]\right\}f(t) \; .
 \label{A-fT-expand}
\end{align}

On the static case, the terms making up $\delta_{j}$ were written as combinations of the Fourier transform of a smooth function. This is no longer valid in the expanding case. We now have
\begin{align*}
||B_{j}^{(n)}||^{2} &=\int\textrm{d}t\, B_{j}^{(n)}(t)^{2} \\
&=\frac{1}{2}\int\textrm{d}t\, f(t)^{2}\left[\textrm{Re}\left((\alpha_{j}^{(n)}(t)+\overline{\beta}_{j}^{(n)}(t))^{2}W_{j}^{(n)}(t)^{2}\right)+\lvert\alpha_{j}^{(n)}(t)+\overline{\beta}_{j}^{(n)}(t)\rvert^{2}\lvert W_{j}^{(n)}(t) \rvert^{2}\right]
\end{align*}
and
\begin{align*}
||D_{j}^{(n)}||^{2} &=\int\textrm{d}t\, D_{j}^{(n)}(t)^{2} \\
&=\frac{1}{2}\int\textrm{d}t\, f(t)^{2}\left[\lvert\alpha_{j}^{(n)}(t)+\overline{\beta}_{j}^{(n)}(t)\rvert^{2}\lvert W_{j}^{(n)}(t) \rvert^{2}-\textrm{Re}\left((\alpha_{j}^{(n)}(t)+\overline{\beta}_{j}^{(n)}(t))^{2}W_{j}^{(n)}(t)^{2}\right)\right]
\end{align*}
Taking into account \eqref{alpha-beta-lambda} and the estimates below equations \eqref{fW.fW} and \eqref{fW*.fW}, we get
\begin{equation*}
 \delta_{j}=\mathcal{O}(\lambda_{j}^{-n-1/2}) \; .
\end{equation*}
The pre-factor of the first term in \eqref{A-fT-expand} is of order
\[(\delta_{j})^{2}=\mathcal{O}(\lambda_{j}^{-2n-1}) \; ,\]
while the one of the second term,
\[(\delta_{j})(2-\delta_{j})=\mathcal{O}(\lambda_{j}^{-n-1/2}) \; .\]
This last one imposes more stringent restrictions. We then obtain for the Hilbert-Schmidt norm of $R_s$
\begin{equation}
||R_s||^2_2\le\sum_j (1+\lambda_j)^{2s-2n-3/2} \; .
\end{equation}  
For the Laplacian on a compact Riemannian space of dimension $m$ we know from Weyl's estimate \cite{Jost11} that $\lambda_j$ is bounded by some constant times $j^{\frac{2}{m}}$. Hence the Hilbert-Schmidt norm of $R_s$ is finite if $s<n+\frac{3}{4}-\frac{m}{4}$.

The modified S-J states are independent of the order of the adiabatic approximation. They thus have the same Sobolev wavefront sets as Hadamard states for every index $s$ and therefore fulfill the Hadamard condition.
\end{proof}

We remark further that for $f\notin \mathcal{C}_{0}^{\infty}(\M)$, the proof that the scalar product \eqref{fW*.fW} decays faster than any power of $\lambda_{j}$ breaks down, and thus there could be some $s$ for which the operator $R_s$ would not be Hilbert-Schmidt, thus $\omega_{SJ_{f}}$ would not be a Hadamard state.

\section{Conclusions}

We propose a new class of states of a free scalar field on globally hyperbolic spacetimes which arise from a variation of the proposal of Sorkin and Johnston. We tested this idea in a class of spacetimes and proved that these states are well defined pure Hadamard states. They are, however, in contrast to the S-J states, not uniquely associated to the spacetime. Several interesting questions might be posed.

First one would like to generalize the construction to generic hyperbolic spacetimes which are relatively compact subregions of another spacetime. This involves some technical problems but we do not see an unsurmountable obstruction. In good cases these states (as also the original S-J states) might converge to a Hadamard state as the subregion increases and eventually covers the full larger spacetime. Such a situation occurs in static spacetimes, and it would be interesting to identify the properties of a spacetime on which this procedure works. There is an interesting connection to the proposal of the fermionic projector of Finster \cite{Finster11} where an analogous construction for the Dirac field was considered. The case of the scalar field is however much easier because of the Hilbert space structure of the functions on the manifold, in contrast to the indefinite scalar product on the spinor bundle of a Lorentzian spacetime.    

Another interesting question concerns the physical interpretation. We do not expect that these states should be interpreted as some kind of vacuum, but we would like to better understand the relation of these states with the States of Low Energy. As a first step one may try to numerically analyze the energy momentum tensor in these states, similarly to the work of Degner on States of Low Energy \cite{Degner09}.

\ack

Marcos Brum would like to thank the AQFT group at the University of Hamburg for the hospitality during the period of preparation of this work. MB also acknowledges financial support from the brazilian agency CAPES under Grant nr 869211-4.

\appendix

\section{Proof of essential self-adjointness of $K$}
\renewcommand{\thesection}{A}
\setcounter{section}{1}

\begin{ATheorem}\label{theorem_essential_self-adjointness}
Let $\M=\mathbb{R}\times \Sigma$ be a globally hyperbolic spacetime with metric $g=\alpha^2(dt^2-h)$ where $h$ is a Riemannian metric on the manifold $\Sigma$ and $\alpha$ a smooth nowhere vanishing function on $\Sigma$. Then
\begin{enumerate}
\item $(\Sigma,h)$ is a complete metric space.\label{a}
\item \label{b} The d'Alembertian  on $\M$  is of the form
\[\square_g=\alpha^{-2}(\partial_{t}^{2}+K)\]
where $K$ is an elliptic differential operator on $\Sigma$ which is positive and selfadjoint on $L^2(\Sigma,\alpha^{-2}\sqrt{|\mathrm{det}g|})$ with core $\mathcal{C}^{\infty}_0(\Sigma)$. 
\end{enumerate}
\end{ATheorem}
\begin{proof}
We follow the papers of Chernoff \cite{Chernoff73} and Kay \cite{Kay78}. 

\eqref{a} The spacetime $\M$ is conformally equivalent to an ultrastatic spacetime with metric $dt^2-h$. The latter is globally hyperbolic if and only if $(\Sigma,h)$ is complete \cite{Kay78}. Hence the same holds true for $\M$.

\eqref{b} In local coordinates, $K$ assumes the form
\[K=-\alpha^2\gamma^{-1}\partial_j\gamma \alpha^{-2}h^{jk}\partial_k\]
with $\gamma=\sqrt{|\mathrm{det}g|}$. The principal symbol of $K$ is $\sigma_K=h^{jk}\xi_j\xi_k$, hence $K$ is elliptic. Moreover, on $L^2(\Sigma,\gamma \alpha^{-2})=:\mathcal{H}$ we have, for $\phi,\psi\in\mathcal{C}^{\infty}_0(\Sigma)$
\begin{align*}
\langle \phi,K\psi\rangle=-\int dx\,\overline{\phi(x)}\partial_j\gamma \alpha^{-2}h^{jk}\partial_k\psi(x)=\int dx\, \gamma \alpha^{-2}h^{kj}\overline{\partial_j\phi(x)}\partial_k\psi(x)=\langle K\phi,\psi\rangle \ ,
\end{align*}
hence $K$ is symmetric and positive, thus one can form the Friedrichs extension and obtains a selfadjoint positive operator.

It remains to prove that $K$ is essentially selfadjoint on $\mathcal{C}^{\infty}_0(\Sigma)$. For this purpose we use a variation of the method of Chernoff and exploit the fact that the Cauchy problem for normally hyperbolic differential equations is well posed on globally hyperbolic spacetimes. 

Let $V(t)$ denote the operator on $\mathcal{D}:=\mathcal{C}^{\infty}_0(\Sigma)\oplus \mathcal{C}^{\infty}_0(\Sigma) $ defined by
\[V(t)\left(\begin{array}{c}\phi_1\\ \phi_2\end{array}\right)=\left(\begin{array}{c}\phi(t)\\ \dot{\phi}(t)\end{array}\right)\] 
where $t\mapsto \phi(t)$ is the solution of the Cauchy problem with initial conditions $\phi(0)=\phi_1$ and $\dot{\phi}(0)=\phi_2$.
The 1-parameter group $t\mapsto V(t)$ satisfies the differential equation
\[\frac{d}{dt}V(t)=iAV(t)\]
with
\[iA=\left(\begin{array}{cc}0 & 1\\-K & 0\end{array}\right)\ . \]
We now equip $\mathcal{D}$ with a positive semidefinite scalar product such that $V(t)$ becomes unitary. We set
\[\langle \phi,\psi\rangle=\langle \phi_1,K\psi_1\rangle+\langle\phi_2,\psi_2\rangle\]
where on the right hand side we use the scalar product of $\mathcal{H}=L^2(\Sigma,\gamma \alpha^{-2})$. (The first component of the scalar product vanishes for $m^2=0$ on functions which are constant on every connected component of $\Sigma$, hence if $\Sigma$ has compact connected components the scalar product is not definite). 
Then
\[\frac{d}{dt}\langle V(t)\phi,V(t)\psi\rangle=0\]
which implies that $V(t)$ is unitary.   

We proceed similarly to the proof of Chernoff's Theorem \cite{Chernoff73}. Let $\psi\in \mathcal{H}=L^2(\Sigma,\gamma\alpha^{-2})$ such that for all $\phi\in\mathcal{C}^{\infty}_0(\Sigma)$ 
\[\langle K\phi,\psi\rangle=i\langle \phi,\psi\rangle \ .\]
We consider the function
\[f:t\mapsto \left\langle V(t)\left(\begin{array}{c}0\\ \phi\end{array}\right),\left(\begin{array}{c}0\\ \psi\end{array}\right)\right\rangle\]
Due to the unitarity of $V(t)$, this function is bounded. By the assumption on $\psi$, it satisfies the differential equation
\[\frac{d^2}{dt^2}f(t)=-\left\langle K\left(V(t)\left(\begin{array}{c}0\\ \phi\end{array}\right)\right)_2,\psi\right\rangle=-i\left\langle \left(V(t)\left(\begin{array}{c}0\\ \phi\end{array}\right)\right)_2,\psi\right\rangle=-if(t)\ .\]
But the only bounded solution of this equation vanishes, hence $\psi$ is orthogonal to $\mathcal{C}^{\infty}_0(\Sigma)$ on the Hilbert space $\mathcal{H}$, hence $\psi=0$. The same argument holds if we replace $i$ by $-i$ in the defining condition on $\psi$.
This proves that $K$ is essentially selfadjoint on $\mathcal{H}$.

\end{proof}

\section*{References}

\end{document}